\documentclass[runningheads]{llncs}
\usepackage{graphicx}
\usepackage{comment}
\usepackage{amsmath}
\usepackage{amsfonts}
\usepackage{amssymb}
\usepackage{float}
\usepackage{subfiles}
\usepackage{hyperref}
\usepackage{xr}
\usepackage{color}
\usepackage{thmtools,thm-restate}
\allowdisplaybreaks
\begin{document}
\title{Bitcoin Security under Temporary Dishonest Majority}
\author{Georgia Avarikioti \and Lukas K\"appeli \and Yuyi Wang \and Roger Wattenhofer}
\authorrunning{G. Avarikioti et al.}
\institute{ETH Zurich, Switzerland \\
\email{\{zetavar,yuwang,wattenhofer\}@ethz.ch,\\ lukas.kaeppeli@hotmail.com}}
\maketitle              % typeset the header of the contribution
%
%----------ABSTRACT------------
%
\begin{abstract}
% Standard approaches for analysing Byzantine Agreement and Consensus Protocols typically just consider two types of players, honest and corrupted ones. These approaches are also used to provide security guarantees for Bitcoin, by assuming an upper bound on the fraction of corrupted players. This assumption doesn't reflect the reality, since it implies that every honest player is online during the whole protocol execution.

% We provide a way to relax this assumption. By dividing the honest parties into alert (online) and sleepy (offline) parties, we can prove the security of Bitcoin under the relaxed assumption that we only require the expected number of alert parties to be larger than the corrupted players (by some factor). This setting allows even temporary dishonest majorities during the protocol execution. We will prove the security of the Bitcoin protocol by proving the three fundamental blockchain properties \textsf{Chain-Growth, Common-Prefix} and \textsf{Chain-Quality}. In the first part of this paper, we are going to prove the three properties for a synchronous model. We then extend these results to the bounded-delay model and to a synchronous model, where we allow message losses. We provide remarkable results for the security of Bitcoin in all three models, especially in the one, where we allow message losses.
We prove Bitcoin is secure under temporary dishonest majority.
We assume the adversary can corrupt a specific fraction of parties and also introduce crash failures, i.e., some honest participants are offline during the execution of the protocol. We demand a majority of honest online participants on expectation. We explore three different models and present the requirements for proving Bitcoin's security in all of them: we first examine a synchronous model, then extend to a bounded delay model and last we consider a synchronous model that allows message losses.

\end{abstract}

%
%----------INTRODUCTION------------
%
\section{Introduction}

Bitcoin~\cite{Nak08} is the predominant cryptocurrency today. Nevertheless, our  understanding of Bitcoin' s correctness is limited. Only relatively recently, there have been attempts to formally capture Bitcoin' s security properties. In a seminal work, Garay et al.~\cite{GKL15} proposed a formal framework (the ``backbone protocol'') to describe the Bitcoin system. They defined security properties for the backbone protocol and proved these both in the synchronous and bounded-delay model. 

%Our work extends the work of Garay et al.~\cite{GKL15} in several dimensions. First, in contrast to our model, \cite{GKL15} assumed a constant honest majority of the participants. We believe that the Bitcoin protocol is more fault-tolerant, and should also be able to allow for a majority of dishonest players, as long as this dishonest majority is only temporary. In this work, we extend the original work of Garay et al.~\cite{GKL15} and prove that this intuition is correct, and Bitcoin only needs an honest majority in expectation.

Our work extends the work of Garay et al.~\cite{GKL15} in several dimensions. First, in contrast to our model, \cite{GKL15} assumed a constant honest majority of the participants. However, the Bitcoin protocol has been proven to be more fault-tolerant and able to allow for a majority of dishonest players, as long as this dishonest majority is only temporary. Specifically, in 2014 there was a majority takeover (approximately 54\% of the network) by the mining pool GHash.io. The cost to perform such attacks have been studied in \cite{bonneau2018hostile} and \url{https://www.crypto51.app/}. In this work, we extend the original work of Garay et al.~\cite{GKL15} to capture these attacks, by allowing a temporary dishonest majority. We provide a formal analysis and investigate under which circumstances Bitcoin is secure when the honest majority holds only on expectation.

Second, motivated by a model of Pass and Shi \cite{PS17}, we not only have honest (``alert'') or dishonest (``corrupted'') nodes. Instead, there is a third group of nodes that are currently not able to follow the protocol. We call them ``sleepy'', which really is a euphemism for  nodes that are basically offline, eclipsed from the action, for instance by a denial of service attack. Understanding this trade-off between corrupted and sleepy nodes gives us a hint whether a dishonest attacker should rather invest in more mining power (to get more corrupted players) or in a distributed denial of service architecture (to get more sleepy players).

Third, we introduce a parameter $c$ that upper bounds the mining power of the adversary over the mining power of ``alert'' nodes. This is not necessary for the security analysis. However, as showed in \cite{SSZ16}, if the adversary follows a selfish mining strategy, he can gain a higher fraction of blocks (rewards) compared to his fraction of the mining power. Hence, parameter $c$ allows us to clearly capture the correlation between this parameter and the advantage of the adversary when he deviates from the honest protocol execution.

Fourth, we study network delays since they significantly affect the performance and security. By extending our synchronous model to the a semi-synchro\-nous model, we show that the upper bound on sleepy parties heavily depends on the maximum allowed message delay.
	
Finally, we extend our analysis to a synchronous model, where we allow message losses. This is inspired by an idea described in \cite{HKZG15eclipse}, where the adversary may perform an eclipse attack \cite{SNDW06,SM02} on some victims  which enables the adversary to control their view of the blockchain. We  show security under the assumption that the adversary can eclipse a certain number of players, depending on the number of corrupted players.

	The omitted theorems, lemmas and proofs can be found in the full version of the paper \cite{avarikioti2019bitcoin}.

\section{The Model}

    We adapt the model originally introduced by Garay et al.~\cite{GKL15} to prove the security of the Backbone protocol. We initially present all the components of a general model and then parametrize the model to capture the three different models under which we later prove that the backbone protocol is secure. 
    
		\subsection{The Execution}

			We assume a fixed set of $n$ parties, executing the Bitcoin backbone protocol. Each party can either be \textit{corrupted}, \textit{sleepy} or \textit{alert};  \textit{sleepy} is an offline honest node and \textit{alert} an honest node that is actively participating in the protocol.

			\subsubsection{Involved programs.}		
			All programs are modeled as polynomially-bounded interactive Turing machines (ITM) that have communication, input and output tapes. An ITM instance (ITI) is an instance of an ITM running a certain program or protocol. 
			Let the ITM $\mathcal{Z}$ denote the environment program that leads the execution of the Backbone protocol. Therefore $\mathcal{Z}$ can spawn multiple ITI's running the protocol. These instances are a fixed set of $n$ parties, denoted by $P_1, \dots, P_n$. 
			The control program $C$, which is also an ITM, controls the spawning of these new ITI's and the communication between them. Further, $C$ forces the environment $\mathcal{Z}$ to initially spawn an adversary $\mathcal{A}$. The environment will then activate each party in a round-robin way, starting with $P_1$. This is done by writing to their input tape. Each time, a corrupted party gets activated, $\mathcal{A}$ is activated instead. The adversary may then send messages $(\mathtt{Corrupt}, P_i)$ to the control program and $C$ will register the party $P_i$ as corrupted, as long as there are less than $t < n$ parties corrupted. Further, the adversary can set each party asleep by sending a message $(\mathtt{sleep}, P_i)$ to the control program. The control program $C$ will set the party $P_i$ asleep for the next round with probability $s$, without informing $\mathcal{A}$ if the instruction was successful or not.
			
			Each party $P_i$ has access to two ideal functionalities, the ``random oracle'' and the ``diffusion channel'', which are also modelled as ITM's. These functionalities, defined below, are used as subroutines in the Backbone protocol.
			
			\subsubsection{Round.}
				A round of the protocol execution is a sequence of actions, performed by the different ITI's. In our setting, a round starts with the activation of the party $P_1$, which then performs the protocol-specific steps. By calling the below defined diffuse functionality, $P_1$ has finished it's actions for the current round and $\mathcal{Z}$ will activate $P_2$. If the party $P_i$ is corrupted, $\mathcal{A}$ will be activated and if $P_i$ is asleep, $P_{i+1}$ gets activated instead. The round ends after $P_n$ has finished.
				Rounds are ordered and therefore enumerated, starting from $1$.

			\subsubsection{Views.}
				Let us formally define the view of a party $P$. The only ``external'' input to the protocol is the security parameter $\kappa$. Therefore, we can consider $\kappa$ to be constant over all rounds of the execution and we can exclude it from the random variable describing the view of a party. 
				We denote by the random variable $VIEW_{\mathcal{A}, \mathcal{Z}}^{P, t, n}$ the view of a party $P$ after the execution of the Bitcoin backbone protocol in an environment $\mathcal{Z}$ and with adversary $\mathcal{A}$.
				The complete view over all $n$ parties is the concatenation of their views, denoted by the random variable $VIEW_{\mathcal{A}, \mathcal{Z}}^{t, n}$.

		\subsubsection{Communication and ``hashing power''.}
		\label{section:functionalities-section}
			The two ideal functionalities, which are accessible by the parties, model the communication between them and the way of calculating values of a hash function $H(\cdot): \{0,1\}^* \rightarrow \{0,1\}^{\kappa} $ concurrently. 
			
			\subsubsection{The random oracle functionality.}
				The random oracle (RO) provides two functions, a calculation and a verification function. Each party is given a number of $q$ calculation queries and unlimited verification queries per round. Thus, an adversary with $t$ corrupted parties may query the random oracle for $t \cdot q$ calculation queries per round.
				Upon receiving a calculation query with some value $x$ by a party $P_i$, the random oracle checks, whether $x$ was already queried before. If not, the RO selects randomly $y \in \{0,1\}^{\kappa}$ and returns it. Further, the RO maintains a table and adds the pair $(x,y)$ into this table. If $x$ was already queried before, the RO searches in the table for the corresponding pair and returns the value $y$ from it. 
				It's easy to see that a verification query now only returns true/valid, if such a pair exists in the table of the RO.
				Note that the RO can maintain tables for different hash functions and can be used for all hash functions we need.
					
			\subsubsection{The diffuse functionality.}
				The diffuse functionality models the communication between the parties and thus maintains a $\mathit{RECEIVE()}$ string for each party $P_i$. Note that this is not the same as the previously mentioned input tape. Each party can read the content of its $\mathit{RECEIVE()}$ string at any time. The message delay is denoted by $\Delta$, where $\Delta = 0$ corresponds to a synchronous setting.
				
				The diffuse functionality has a $round$ variable, which is initially set to 1. Each party $P_i$ can send a message $m$, possibly empty, to the functionality, which then marks $P_i$ as complete for the current round. We allow $\mathcal{A}$ to read all the messages that are sent by some $P_i$, without modifying, dropping or delaying it. When all parties and the adversary are marked as complete, the functionality writes all messages that are $\Delta$ rounds old to the $\mathit{RECEIVE()}$ strings of either only the alert or all parties. We denote by $B$ a Boolean function that indicates exactly that; if $B=0$ the diffuse functionality  writes all messages to the $\mathit{RECEIVE()}$ strings of the alert parties, while if $B=1$ the diffuse functionality  writes all messages to the $\mathit{RECEIVE()}$ strings of all parties. Each party can read the received messages in the next round being alert. 
				At the end, $round$ is incremented.
				
				Note that in the case where $B=1$, if a party is asleep at a round, it automatically gets marked as complete for this round. Further, upon waking up, it can read all the messages that were written to its $\mathit{RECEIVE()}$ string while it was asleep.

            \subsubsection{Successful queries.}
		        A query to the RO oracle is successful, if the returned value $y < T$, where $T$ is the difficulty parameter for the PoW function. The party, which have issued the query will then create a new valid block and may distribute it by the diffuse functionality. We denote the success probability of a single query by $p = Pr[y < T] = \frac{T}{2^\kappa}$. Note that in Bitcoin, the difficulty parameter is adjusted such that the block generation time is approximately ten minutes.

		\subsection{Sleepy, Alert and Corrupted}
			For each round $i$, we have at most $t$ corrupted and $n_{honest, i} = n - t$ honest parties. Furthermore, the number of honest parties are divided to  alert and sleepy parties, $n_{honest, i} = n_{alert, i} + n_{sleepy, i}$. We assume without loss of generality that no corrupted party is asleep, since we only upper-bound the power of the adversary.
			Since $n_{alert, i}$ and $n_{sleepy, i}$ are random variables, we can also use their expected value. The expected value is constant over different rounds, thus we will refer to them as $E[n_{alert}]$ and $E[n_{sleepy}]$. Since each honest party is independently set to sleep with probability $s$ and thus the random variable $n_{sleepy, i}$ is binomially distributed with parameters $(n-t)$ and $s$. Accordingly, $n_{alert, i}$ is also binomially distributed with parameters $(n-t)$ and $(1-s)$. Hence,	$E[n_{sleepy}] = s \cdot (n-t)$ and  $E[n_{alert}] = (1-s) \cdot (n-t)$.

    \subsection{Parametrized Model}
        Let $M(q, \Delta, B)$ be the model, defined in this section. In the following sections, we will look at three instantiations of this model. First of all, we are going to analyze the model $M(q,0,1)$, which corresponds to a synchronous setting, in which each party has the ability to make $q$ queries to the random oracle and receives every message, even if the party is asleep. 
        Then, we extend these results to the bounded delay model, which corresponds to $M(1,\Delta,1)$. As before, every party will always receive messages, but we restrict $q$ to be $1$. In the last section, we analyze the model $M(q,0,0)$, which corresponds to the synchronous model, but we do not allow the diffuse functionality to write messages on the $RECEIVE()$ tapes of sleepy parties.

    \subsection{Properties}	
	    In order to prove the security of the Bitcoin backbone protocol, we are going to analyze three different properties, following the analysis of \cite{GKL15}. These properties are defined as predicates over $VIEW_{\mathcal{A}, \mathcal{Z}}^{t, n}$, which will hold for all polynomially bounded environments $\mathcal{Z}$ and adversaries $\mathcal{A}$ with high probability.
				
		\begin{definition}
		\label{definition:predicate-definition}
			Given a predicate $Q$ and a bound $q, t, n \in \mathbb{N}$ with $t < n$, we say that the Bitcoin backbone protocol satisfies the property $Q$ in the model $M(q, \Delta, B)$ for $n$ parties, assuming the number of corruptions is bounded by $t$, provided that for all polynomial-time $\mathcal{Z}, \mathcal{A}$, the probability that $Q(VIEW_{\mathcal{A}, \mathcal{Z}}^{t, n})$ is false is negligible in $\kappa$.
		\end{definition}
				
		The following two Definitions concern the liveness and eventual consistency properties of the Backbone protocol. We are using the notation of \cite{GKL15}: We denote a chain $C$, where the last $k$ blocks are removed, by $C^{\lceil k}$. Further, $C_1 \preceq C_2$ denotes that $C_1$ is a prefix of $C_2$. 
				
		\begin{definition}
		\label{definition:chain-growth-property-definition}
			The chain growth property $\mathit{Q}_{cg}$ with parameters $\tau \in \mathbb{R}$ and $s \in \mathbb{N}$ states that for any honest party $P$ with chain $\mathit{C}$ in $VIEW^{t,n}_{ \mathcal{A}, \mathcal{Z}}$, it holds that for any $s+1$ rounds, there are at least $\tau \cdot s$ blocks added to the chain of $P$. 
			\footnote{The Chain-Growth Property in \cite{GKL15} is defined slightly different: \textit{\dots it holds that for any $s$ rounds, there are at least $\tau \cdot s$ blocks added to the chain of $P$}. Considering the proof for Theorem 1 (of \cite{avarikioti2019bitcoin}), one can see, why we use $s+1$ instead of $s$. It follows by the fact that the sum in Lemma 13 (of \cite{avarikioti2019bitcoin}) only goes from $i=r$ to $s-1$ and not to $s$.}
		\end{definition}
		
		\begin{definition}
		\label{definition:common-prefix-definition}
			The common-prefix property $Q_{cp}$ with parameter $k \in \mathbb{N}$ states that for any pair of honest players $P_1, P_2$ adopting the chains $C_1, C_2$ at rounds $r_1 \leq r_2$ in $VIEW^{t,n}_{\mathcal{A}, \mathcal{Z}}$ respectively, it holds that $C_1^{\lceil k}\preceq C_2$.
		\end{definition}
			
		In order to argue about the number of adversarial blocks in a chain, we will use the chain quality property, as defined below:	
			
        \begin{definition}
		\label{definition:chain-quality-definition}
			The chain quality property $Q_{cq}$ with parameters $\mu \in \mathbb{R}$ and $l \in \mathbb{N}$ states that for any honest party $P$ with chain $C$ in $VIEW^{t,n}_{\mathcal{A}, \mathcal{Z}}$, it holds that for any $\ell$ consecutive blocks of $C$ the ratio of adversarial blocks is at most $\mu$.
		\end{definition}
		
		The following two definitions formalize typical executions of the Backbone protocol. Both of them are related to the hash functions, used for implementing the Backbone Protocol. Further, the parameters $\epsilon$ and $\eta$ are introduced. Throughout the paper, $\epsilon \in (0,1)$ refers to the quality of concentration of random variables in typical executions and $\eta$ corresponds to the parameter, determining block to round translation.
		        
		\begin{definition}[\cite{GKL15}, Definition 8]
			An \textit{insertion} occurs when, given a chain $\mathcal{C}$ with two consecutive blocks $B$ and $B'$, a block $B^*$ is such that $B, B^*, B'$ form three consecutive blocks of a valid chain. A \textit{copy} occurs if the same block exists in two different positions. A \textit{prediction} occurs when a block extends one which was computed at a later round.
		\end{definition}

		\begin{definition}[\cite{GKL15}, Definition 9]
		\label{definition:typical-execution-definition}
			\textit{(Typical execution)}. An execution is $(\epsilon, \eta)-typical$ if, for any set $S$ of consecutive rounds with $|S| \geq \eta\kappa$ and any random variable $X(S)$, the following holds:
			\begin{itemize}
				\item[a)] $(1-\epsilon) E[X(S)] < X(S) < (1+\epsilon) E[X(S)]$
				\item[b)]No insertions, no copies and no predictions occurred.
			\end{itemize}
		\end{definition}
			
		\begin{lemma}
		\label{theorem:typical-execution-theorem}
			An execution is typical with probability $1-e^{-\Omega(\kappa)}$. 
		\end{lemma}
			
		\begin{proof} 
		    To prove a), we can simply use a Chernoff bound by arguing that $E[X(S)]$ is in $\Omega(|S|)$. The proof for b) is equivalent to \cite{GKL15}, by reducing these events to a collision in one of the hash functions of the Bitcoin backbone protocol. Such collisions only happen with probability $e^{-\Omega(\kappa)}$.\qed
				
		\end{proof}

%---------------------------------------------------------------------------------------------
%
%   SECTION:    THE Q-BOUNDED SYNCHRONOUS MODEL WITHOUT MESSAGE LOSS
%
%---------------------------------------------------------------------------------------------
    
\section{The $q$-bounded Synchronous Model without Message Loss $M(q,0,1)$}
        In this section, we analyze the Bitcoin backbone protocol in the previously defined model, instantiated as $M(q,0,1)$. This corresponds to the $q$-bounded synchronous setting in \cite{GKL15}.
        First, we define the success probabilities for the alert and corrupted parties, which are used to prove the relations between them. At the end, we use these results to show the properties of chain growth, common prefix and chain quality. 
        
        Following the definition in \cite{GKL15}, let a successful round be a round in which at least one honest party solves a PoW. The random variable $X_i$ indicates successful rounds $i$ by setting $X_i = 1$ and $X_i = 0$ otherwise. Further, we denote for a set of rounds $S$: $X(S) = \sum_{i \in S} X_i$. We note that if no party is asleep, we have $E[X_i] = Pr[X_i = 1] = 1- (1-p)^{q(n-t)}$.
		
%--------------LESS IMPORTANT LEMMA (But the proof deviates from the original work since the introduction of the random variable) -------------------------------------
	\newpage	 
	    \begin{restatable}{lemma}{successfulrounds}
		\label{lemma:successful-rounds-lemma}
		    It holds that $\frac{pqE[n_{alert}]}{1 + pqE[n_{alert}]}  \leq E[X_i] \leq pqE[n_{alert}]$.
		\end{restatable}
				
		\begin{proof}
		    By the definition of $X_i$, we know that $E[X_i] = E[1 - (1-p)^{qn_{alert,i}}]$. Thus, the second inequality can easily be derived using Bernoulli. And for the first inequality holds:
				    
			\begin{align*}
			    E[X_i]  &= \sum_{k=0}^{n-t}E[X_i|n_{alert,i}=k] \cdot Pr[n_{alert,i} = k] \\
				&= \sum_{k=0}^{n-t}\Big(1 - (1-p)^{qk}\Big) \cdot \binom{n-t}{k}(1-s)^k s^{n-t-k} \\
				&= 1 - \Big(s - (s-1)(1-p)^q \Big)^{n-t}
				\geq 1 - \Big(s - (s-1)(1-pq)\Big)^{n-t} \\
				&\geq 1 - e^{-(1-s)(n-t)pq}
				= \frac{pqE[n_{alert}]}{1 + pqE[n_{alert}]}
			\end{align*}
		\qed\end{proof}
				
        We also adapt the notation of a unique successful round from \cite{GKL15}. A round is called a unique successful round, if exactly one honest party obtains a PoW. Accordingly to the successful rounds, let the random variable $Y_i$ indicates a unique successful round $i$ with $Y_i = 1$ and $Y_i = 0$ otherwise. And for a set of rounds $S$, let $Y(S) = \sum_{i \in S} Y_i$.

        \begin{restatable}{lemma}{uniquesuccess}
		\label{lemma:unique-successful-rounds-lemma}
		    It holds $E[Y_i] = E[pqn_{alert,i}(1-p)^{q(n_{alert,i}-1)}] \geq E[X_i](1-E[X_i])$.
    	\end{restatable}
		
		\begin{proof}
		    To prove the required bounds, we need a few intermediary steps. Using Bernoulli, we can derive the following:

			$$E[Y_i] = E[pqn_{alert,i}(1-p)^{q(n_{alert,i}-1)}] \geq E[pqn_{alert,i}(1-pq(n_{alert, i} -1))]$$ 
		
			Then, we have to prove that $pqE[n_{alert}](1- pqE[n_{alert}]) \geq E[X_i](1-E[X_i])$. From the upper bound on $E[X_i]$, we can derive $E[X_i] = pqE[n_{alert}] - b$, for $b \geq 0$. Therefore:
				
			\begin{align*}
			    E[X_i] (1-E[X_i])   &= (pqE[n_{alert}] -b)(1 - pqE[n_{alert}] + b) \\
									&= pqE[n_{alert}](1- pqE[n_{alert}]) - b^2 -b + 2pqE[n_{alert}]b
			\end{align*}
				
			In order to prove the required bound, it must hold that $0 \geq -b^2 -b + 2pqE[n_{alert}]b$, which is equivalent to $1 \geq E[X_i] + pqE[n_{alert}]$ and holds by the fact that $2E[X_i] \leq 1$. This is also required by the proof in \cite{GKL15}, but not stated explicitly. Since in Bitcoin, $E[X_i]$ is between $2\%-3\% $, the inequality can be justified. 
				
			To conclude the proof, we just have to prove the following:
				
			\begin{align*}
				&E[pqn_{alert,i}(1-pq(n_{alert, i} -1))] &&\geq pqE[n_{alert}] - (pq)^2E[n_{alert}]^2\\
				\Leftarrow& E[n_{alert}{^2}] - E[n_{alert}] && \leq E[n_{alert}]^2 \\
			\end{align*}
				
			Which is equivalent to $Var[n_{alert}] \leq E[n_{alert}]$ and holds for the binomial distribution.
		\qed\end{proof}

        % The proofs of both Lemma \ref{lemma:successful-rounds-lemma} and Lemma \ref{lemma:unique-successful-rounds-lemma} can be found in Appendix \ref{app:Z}.
        
        Let the random variable $Z_{ijk} = 1$ if the adversary obtains a PoW at round $i$ by the $j^{th}$ query of the $k^{th}$ corrupted party. Otherwise, we set $Z_{ijk} = 0$. Summing up, gives us $Z_i = \sum_{k=1}^t \sum_{j=1}^q Z_{ijk}$ and $Z(S) = \sum_{i \in S} Z_i$. Then, the expected number of blocks that the adversary can mine in one round $i$ is:
				
		$$E[Z_i]=qpt=\frac{t}{E[n_{alert}]} pqE[n_{alert}]\leq \frac{t}{E[n_{alert}]}\cdot \frac{E[X_i]}{1-E[X_i]}$$
		
		\subsection{Temporary Dishonest Majority Assumption}\label{dishonest-synchronous}
		    We assume the honest majority assumption holds on expectation. In particular, for each round the following holds: 
			%A number of $t$ out of $n$ parties are corrupted such that 
				$t \leq c \cdot (1- \delta)\cdot E[n_{alert}]$,
			where  $\delta \geq 2 E[X_i] + 2 \epsilon$ and $c \in [0,1]$ is a constant. As in \cite{GKL15}, $\delta$ refers to the advantage of the honest parties and $\epsilon$ is defined in Definition \ref{definition:typical-execution-definition}.
				
		%\subsubsection{Relation between t and s.}
			From the expected honest majority assumption, we can derive a possible upper bound for $s$, depending on $t, \delta$ and $c$. Formally,
			
			$$ s \leq \frac{n - t - \frac{t}{c(1-\delta)}}{n-t} = 1 - \frac{1}{c(1-\delta)}\frac{t}{n-t}$$
			
%			\begin{figure}
%			\label{figure:synchronous}
%			\centering
%			\includegraphics[scale=0.6]{max-s-synchronous.eps}
%			\caption{This figure shows the upper bound on $s$, depending on the fraction $\frac{t}{n}$. In order to do that, we set $\delta = 0.06$, which approximately reflects the reality, where $E[X_i]$ is between $2\%-3\%$. Further, we set $c = 0.5$, which upper bounds the fraction of the adversarial blocks, as we show later on.}
%			\end{figure}

	\subsection{Security Analysis}	
        First of all, by Definition \ref{definition:typical-execution-definition} the properties of the typical execution hold for the random variables $X(S), Y(S), Z(S)$, assuming $|S| \geq \eta\kappa$.
% 		\begin{definition}[\cite{GKL15}, Definition 8]
% 			An \textit{insertion} occurs when, given a chain $\mathcal{C}$ with two consecutive blocks $B$ and $B'$, a block $B^*$ is such that $B, B^*, B'$ form three consecutive blocks of a valid chain. A \textit{copy} occurs if the same block exists in two different positions. A \textit{prediction} occurs when a block extends one which was computed at a later round.
% 		\end{definition}

% 		\begin{definition}[\cite{GKL15}, Definition 9]
% 		\label{definition:typical-execution-definition}
% 			\textit{(Typical execution)}. An execution is $(\epsilon, \eta)-typical$ if, for any set $S$ of consecutive rounds with $|S| \geq \eta\kappa$, the following holds:
% 			\begin{itemize}
% 				\item[a)] $(1-\epsilon) E[X(S)] < X(S) < (1+\epsilon) E[X(S)]$
% 				\item[b)] $(1-\epsilon) E[Y(S)] < Y(S)$
% 				\item[c)] $Z(S) < (1+\epsilon) E[Z(S)]$
% 				\item[d)] No insertions, no copies and no predictions occurred.
% 			\end{itemize}
% 		\end{definition}
			
% 		\begin{lemma}
% 		\label{theorem:typical-execution-theorem}
% 			An execution is typical with probability $1-e^{-\Omega(\kappa)}$. 
% 		\end{lemma}
			
% 		\begin{proof} 
% 		    To prove a), b) and c), we can simply use a Chernoff bound by arguing that $E[X(S)]$, $E[Y(S)]$ and $E[Z(S)]$ are in $\Omega(|S|)$. The proof for the property d) is equivalent to \cite{GKL15}, by reducing these events to a collision in one of the hash functions of the Bitcoin backbone protocol. Such collisions only happen with probability $e^{-\Omega(\kappa)}$.
				
% 		\qed\end{proof}
				
		The following lemma shows the relations between the different expected values. The bounds are required in all proofs of the three properties and therefore essential.
		\footnote{The statement d) uses different factors as \cite{GKL15}. The problem is, that it's even not possible to prove the bounds from \cite{GKL15} with their theorems, lemmas and assumptions.}
		
		\begin{restatable}{lemma}{typicalexecutiona}
		\label{lemma:typical-execution-lemma}
			The following hold for any set S of at least $\eta\kappa$ consecutive rounds in a typical execution. 
			
			\begin{itemize}
				\item[a)] $(1 - \epsilon) E[X_i] |S| < X(S) <  (1 + \epsilon) E[X_i] |S|$	
					
				\item[b)] $(1 - \epsilon) E[X_i] (1-E[X_i]) |S| < Y(S)$
					
				\item[c)] $Z(S) < (1 + \epsilon) \frac{t}{E[n_{alert}]} \frac{E[X_i]}{1- E[X_i]} |S| \leq c (1 + \epsilon) (1- \delta) \frac{E[X_i]}{1 - E[X_i]} |S|$
					
				\item[d)] For  $\sigma = (1 - \epsilon)(1 - E[X_i])$: 
					$$Z(S) < \Big(1 + \frac{\delta}{\sigma}\Big) \frac{t}{E[n_{alert}]} X(S) \leq c\Big( 1 - \frac{\delta^2}{2\sigma}\Big) X(S)$$
		            
		       	\item[e)] $Z(S) < Y(S)$
				
			\end{itemize}
		\end{restatable}

	Next, we prove Bitcoin is secure under temporary dishonest majority in the $q$-bounded synchronous setting by proving the three properties defined in \cite{GKL15}: \emph{chain growth}, \emph{common prefix} and \emph{chain quality}. The proofs can be found in the full version.

\section{The semi-Synchronous Model without Message Loss $M(1,\Delta,1)$}	
	    In this section, we extend the previously seen results to the semi-synchronous (bounded delay) model. This means, that we allow $\Delta$\footnote{According to Theorem 11 of \cite{PS17}, the parameter $\Delta$ has to be known by the honest parties to achieve state machine replication, e.g. achieving consensus.} delays for the messages, as described in the Definition of our model. In order to realize the proofs, we have to restrict $q$ to be $1$. And as in the last section, we do not assume message losses.

		Due the introduced network delays, we need to redefine unique successful rounds, because they do not provide the same guarantees in the this model. Especially, Lemma 15 (of \cite{avarikioti2019bitcoin}) will not hold in the new model.
		Therefore, we will introduce two new random variables, one for successful and one for unique successful rounds in the bounded delay model. Note, that the chances for the adversary do not change and we can use the bounds from the synchronous model.
		
	    Let the random variable $X'_i$ be defined such that for each round $i$, $X'_i = 1$, if $X_i = 1$ and $X_j = 0$, $\forall j \in \{i - \Delta + 1, \dots , i - 1 \}$. A round $i$ is called $\Delta$-isolated successful round, if $X'_i = 1$. Further, let $X'(S) = \sum_{i \in S} X'_i$. Using Bernoulli, we can derive the following bound on $E[X'_i]$:
	            
	   $$E[X'_i] = E[X_i] (1-E[X_i])^{\Delta - 1} \geq E[X_i](1- (\Delta -1)E[X_i]).$$
	  
		In order to prove eventual consistency, we have to 	rely on stronger events than just uniquely successful rounds. In \cite{GKL15}, this is achieved by defining the random variable $Y'_i$ such that for each round $i$, $Y'_i = 1$, if $Y_i = 1$ and $X_j = 0$, $\forall j \in \{i -\Delta + 1, \dots, i-1, i+1, \dots , i + \Delta - 1\}$. Then, a round $i$ is called $\Delta$-isolated unique successful round, if $Y'_i = 1$. Further, let $Y'(S) = \sum_{i \in S} Y'_i$. As before, we can lower bound $E[Y'_i]$ using Bernoulli:

        $$E[Y'_i] = E[X_i](1-E[X_i])^{2\Delta -1} \geq E[X_i](1-(2\Delta -1)E[X_i]).$$
	       
		\subsection{Temporary Dishonest Majority Assumption} \label{dishonest-semisynchronous}
			We assume again honest majority on expectation, such that for each round $t \leq c \cdot (1- \delta)\cdot E[n_{alert}]$, where $\delta \geq 2\Delta E[X_i] + 4 \epsilon + \frac{4\Delta}{\eta\kappa}$ and $c \in [0,1]$ is a constant.
			\footnote{One might notice that our lower bound of $\delta$ differs from the lower bound from \cite{GKL15}. First of all, they provided two different values for $\delta$, where both of them are wrong in the sense that they are too small in order to prove the needed bounds. }
			The reason for the higher value of $\delta$ (compared to the synchronous model) is that $E[Y'_i] \leq E[Y_i]$ and we need a way to compensate this difference. 
			
%			\begin{figure}
%			\centering
%			\includegraphics[scale=0.6]{max-s-bounded-delay.eps}
%			\caption{This figure shows the upper bound on $s$, depending on the fraction $\frac{t}{n}$. The upper bound is exactly the same as in the synchronous model without message losses, except that we use another value for $\delta$, as defined above. In this case, we set $\Delta=10$ and $c = 0.5$ and by the Definition of $\delta$, it is equal to $0.44$.}
%			\end{figure}
			
		\subsection{Security Analysis}	
		In this subsection, we prove Bitcoin is secure, i.e. the chain growth, common prefix and chain quality properties hold, for the semi-synchronous model without message loss.
		We note that the properties of the typical execution apply to the predefined random variables $(X'(S), Y'(S), Z(S))$, given that $|S| \geq \eta\kappa$.
%To prove chain growth, common prefix and chain quality, we first define typical executions for this model.
			
%             \begin{definition}[\cite{GKL15}, Definition 27]
% 				\label{definition:bounded-delay-typical-execution-definition}
% 				An execution is $(\epsilon, \eta, \Delta)$-typical if, for any set $S$ of consecutive rounds with $|S| \geq \eta\kappa$, the following hold.
				
% 				\begin{itemize}
% 					\item[a)] $(1-\epsilon)E[X'(S)] < X'(S)$ and $X(S) < (1+\epsilon)E[X(S)]$
% 					\item[b)] $(1-\epsilon)E[Y'(S)] < Y'(S)$
% 					\item[c)] $Z(S) < (1+\epsilon)E[Z(S)]$
% 					\item[d)] No insertions, no copies and no predictions occurred.
% 				\end{itemize}
% 			\end{definition}
			
% 			\begin{theorem}
% 				An execution is typical with probability $1-e^{-\Omega(\kappa)}$.
% 			\end{theorem}
			
% 			\begin{proof}
%                 Equivalent to the proof of Theorem \ref{theorem:typical-execution-theorem}.
% 			\qed\end{proof}
			
			The following lemma corresponds to the semi-synchronous version of Lemma \ref{lemma:typical-execution-lemma}. Most of the relations follow the same structure and are similar to prove as in the synchronous model. 
			
			\begin{restatable}{lemma}{typicalexecutionb}
				\label{lemma:bounded-delay-typical-execution-lemma}
				The following hold for any set S of at least $\eta\kappa$ consecutive rounds in a typical execution.
				
				\begin{itemize}
					\item[a)] $(1-\epsilon)E[X_i](1-E[X_i])^{\Delta-1}|S| < X'(S)$
					\item[b)] $(1-\epsilon)E[X_i](1-E[X_i])^{2\Delta-1}|S| < Y'(S)$
					\item[c)] $Z(S) < (1+\epsilon)\frac{t}{E[n_{alert}]}\frac{E[X_i]}{1-E[X_i]}|S| \leq c (1+\epsilon)(1-\delta)\frac{E[X_i]}{1-E[X_i]}|S|$
					\item[d)] Let $S' = \{r, \dots , r'\}$ with $|S'| \geq \eta\kappa$. For $S = \{r, \dots , r' + \Delta\}$ and $\sigma' = (1-\epsilon)(1-E[X_i])^\Delta$:
					
					$$Z(S) < \Big(1 + \frac{\delta}{2\sigma'}\Big) \frac{t}{E[n_{alert}]} X'(S')$$
					
					\item[e)] Let $S' = \{r, \dots , r'\}$ with $|S'| \geq \eta\kappa$. For $S = \{r - \Delta, \dots , r' + \Delta\}$: $$Z(S) < Y'(S')$$
				\end{itemize}
			\end{restatable}
			
			The proof of Lemma \ref{lemma:bounded-delay-typical-execution-lemma} as well as the proofs of the security properties can be found in the full version.

\section{The $q$-bounded Synchronous Model with Message Loss $M(q,0,0)$}

    As in the synchronous case, we do not restrict the number of queries and assume no message delays. In the previous sections, we assumed that messages, sent from the diffusion functionality, will be written on the $\mathit{RECEIVE()}$ string of each party. However, in this section, we assume that the messages only get written to the $\mathit{RECEIVE()}$ strings of alert parties, i.e. sleepy parties do not receive messages. This models the worst possible event of the reality, because in Bitcoin itself, parties that were offline will check on the currently longest chain, once they get back online. This model captures the effects if none of them receives one of the currently longest chains, thus are eventually a victim of an eclipse attack. This implies that it's not necessarily true that all parties' local chains have the same length.
    
	This change to the model leads to major differences compared to the results from the previous sections. In this case, unique successful rounds doesn't provide the same guarantees as before, especially Lemma 15 (of \cite{avarikioti2019bitcoin}) doesn't hold any more.
    
    In the following, we denote by $C_i$ the set of chains containing all longest chains that exist at round $i$. Further, we refer to the local chain of player $P_j$ at round $i$ by $L_i^j$.
		
	The following lemma shows the expected number of honest players, which have adopted one of the longest chains existing at the current round.
		
%------------ IMPORTANT LEMMA -------------------------------		
		
	\begin{lemma}
	\label{lemma:longest-chain-lemma}
		At every round $i$, there are expected $E[n_{alert}]=(1-s)(n-t)$ parties $j$, such that $L_i^j \in C_i$. 
	\end{lemma}
		
	\begin{proof}
		We will prove the lemma by induction over all rounds of an execution. The base case is trivial, because at round $1$, every party starts with the genesis block.  
		Now for the step case, assume that the lemma holds at round $i$. Then we show that it holds at round $i+1$ too. In order to prove this, we perform a case distinction:
			
		\begin{itemize}
			\item Case $X_i = 0$: No new chains will be diffused, therefore no new chains can be adopted and we     can apply the induction hypothesis. 
			\item Case $Z_i = 0$: Analogue to the previous case.
			\item Case $X_i = 1$: (But $Y_i = 0$) Now we have to differentiate, if the new blocks extend some chain in $C_i$ not:
					\begin{itemize}
						\item[a)] Some longest chain is extended:
						
							Every party, which is not asleep at round $i$ will adopt one of the possibly multiple resulting new longest chains. Thus, there are expected $E[n_{alert}]$ alert parties which will have adopted one of the longest chains at round $i+1$.
						\item[b)] No longest chain is extended:
						
							No honest party, whose local chain is already one of the currently longest chain will adopt a new chain, since it's length will not be larger than the length of its local chain. Thus, we can apply the induction hypothesis.							
					\end{itemize}
				\item Case $Y_i = 1$: As in the case before, every party, which was alert at round $i$, will adopt the resulting chain, if its length is larger than the length of its local chain. As before, there are $E[n_{alert}]$ alert parties which will have adopted one of the longest chains at round $i+1$.
				\item Case $Z_i = 1$: Analogue to the previous case. But if the adversary withholds the found block, the case $Z_i = 0$ applies and at the round, where it diffuses this block, this case applies.\qed
			\end{itemize}
		\end{proof}
		
		By the lemma above, at every round $i$ only expected $(1-s)(n-t)$ parties $j$ have a local chain $L_i^j \in C_i$. And a fraction of $(1-s)$ of them will again be sleepy in the following rounds. Therefore, let $n^*_{alert, i}$ denote the number of alert parties $j$ at round $i$, where $L_i^j \in C_i$. 
		
		It's easy to see that $n^*_{alert, i}$ is binomially distributed with parameters $(n-t)$ and $(1-s)^2$. Let $E[n^*_{alert}] = (1-s)^2(n-t)$ denote the expected value of $n^*_{alert,i}$, omitting the round index $i$, since the expected value is equal for all rounds.
		We define the random variable $X^*_i$ which indicates, if at least one of the $n^*_{alert,i}$ parties solves a PoW at round $i$. Thus, we set $X^*_i = 1$, if some honest party $j$ with $L_i^j \in C_i$ solves a PoW at round $i$ and $X^*_i = 0$ otherwise. Further, we define for a set of rounds S: $X^*(S) = \sum_{i \in S} X^*_i$. 
		%The following Lemma can be proven, using the same argumentation as in the proof for the Lemma \ref{lemma:successful-rounds-lemma}:
			
		\begin{lemma}
		    It holds that $\frac{pqE[n^*_{alert}]}{1 + pqE[n^*_{alert}]} \leq E[X^*_i] \leq pqE[n^*_{alert}] $.
		\end{lemma}
		\begin{proof}
		The lemma can be proven using the same argumentation as in the proof for the Lemma \ref{lemma:successful-rounds-lemma}. \qed
		\end{proof}

		Accordingly, let $Y^*_i$ denote the random variable with $Y^*_i = 1$, if exactly one honest party $j$ solves a PoW at round $i$ and $L_i^j \in C_i$. Note that the resulting chain, will be the only longest chain. Further, for a set of rounds $S$ let $Y^*(S) = \sum_{i\in S} Y^*_i$.	
	    
	    \begin{lemma}
		    It holds $E[Y^*_i] = E[pqn^*_{alert,i}(1-p)^{q(n^*_{alert, i}-1)}] \geq E[X^*_i](1 - E[X^*_i])$.
		\end{lemma}
			
		\begin{proof}
		    The proof follows the exactly same steps as the proof for Lemma \ref{lemma:unique-successful-rounds-lemma}. \qed
			    
%		   \begin{align*}
%					&E[{n^*_{alert}}^2] - E[n^*_{alert}] &&\leq E[n^*_{alert}]^2 \\
%	\Leftrightarrow &(n-t)^2(1-s)^2(2s - s^2) + (1-s)^4(n-t)^2 - (1-s)^2(n-t)^2 && \leq (1-s)^4(n-t)^2\\
%    \Leftrightarrow &(n-t)^2(1-s)^2(2s - s^2) - (1-s)^2(n-t)^2 && \leq 0\\
%	\Leftrightarrow &2s - s^2   &&\leq 1\\
%	\Leftrightarrow	& 0 		&&\leq (1-s)^2
%			\end{align*}
			
%			which holds trivially.  
	    \end{proof}

		\subsection{Temporary Dishonest Majority Assumption}\label{dishonest-messageloss}
			In this setting, the honest majority assumption changes slightly. We cannot simply assume that $t$ is smaller than some fraction of $E[n^*_{alert}]$, because we have also to consider parties $j$ with $L_i^j \notin C_i$. We assume that for each round holds $ t + (1-s)E[n_{sleepy}] \leq c \cdot (1-\delta) \cdot E[n^*_{alert}]$, where $\delta \geq 3\epsilon + 2E[X^*_i]$ and some constant $c \in [0,1] $. Note that $(1-s)E[n_{sleepy}]$ is the fraction of alert parties, working on shorter chains.
			
			In order to compute the upper bound for $s$, we reformulate the honest majority assumption. Using the quadratic formula, this results in the following:
			
			$$ s \leq \frac{2c(1-\delta) - \sqrt{1 + 4(1+c(1-\delta)\frac{t}{n-t})}}{2(1+c(1-\delta))}$$

%			\begin{figure}
%			\centering
%			\includegraphics[scale=0.6]{max-s-msg-losses.eps}
%			\caption{This figure shows the upper bound on $s$, depending on the fraction $\frac{t}{n}$. In order to do that, we fix $\delta = 0.06$ as we did in the last section and set $c = 0.5$.}
%			\end{figure}
			
			In the model description, we specified that the adversary is not informed if a party $P_i$ is set to sleep, after sending an instruction $(\mathtt{sleep}, P_i)$ to the control program $C$. This assumption is realistic since the adversary can not be certain about the success of his attempt to create a crash-failure. Further, allowing the adversary to know when he successfully set to sleep a node makes him quite powerful. Specifically, in our model we have a fraction of $1-s$ alert parties. Subtracting the parties, which are working on a longest chain, from the $(1-s)(n-t)$ parties, leaves us an expected fraction of $s(1-s)$ parties, which can be found on the left hand side of the honest majority assumption. 
			If we would assume that the adversary knows, which parties are asleep at each round, we would have to change the temporary dishonest majority assumption to $ t + E[n_{sleepy}] \leq c \cdot (1-\delta) \cdot E[n^*_{alert}]$. Then, the adversary could exploit this knowledge to his advantage and send sleep instructions to the parties working on the longest chains. To capture this adversarial behavior a different model would be necessary (since $s$ cannot be considered constant).
		
	\subsection{Security Analysis}
	    For this section, we note that the properties of an typical execution apply for the random variables $X^*(S), Y^*(S)$ and $Z(S)$, given that $|S| \geq \eta\kappa$.
	    
%--------LESS IMPORTANT LEMMA, BUT FAST PROVEN-----------------
			\begin{lemma}
			    \label{lemma:chain-growth-message-losses-lemma}
				Suppose that at round $r$, the chains in $C_i$ have size $l$. Then by round $s\geq r$, an expected number of $E[n_{alert}] = (1-s)(n-t)$ parties will have adapted a chain of length at least $l + \sum_{i = r}^{s-1} X^*_i$.
			\end{lemma}
			
			\begin{proof}
				By Lemma \ref{lemma:longest-chain-lemma}, for every round $i$, the expected number of parties $j$ with $L_i^j \in C_i$ is $E[n_{alert}]$. Therefore, we only have to count the number of times, when one of these longest chains gets extended.
			\qed\end{proof}
			
			In the following, we define a new variable $\phi$ and provide an upper bound for it. This is required for the proof of the common prefix property. Although the proven bound is not tight, it is sufficient for proving the desired properties.
			
%-----------------------MOST IMPORTANT LEMMA-----------------------------------
			
			\begin{lemma}
			\label{lemma:small-chance-lemma}
				The probability that the honest parties $j$ with $L_i^j \notin C_i$ can create a new chain $C' \in C_r$ for some round $r \geq i$, before any chain from $C_i$ gets extended is denoted by $\phi$. It holds that: 
			                    $$\phi \leq \frac{s}{1-s} $$
			\end{lemma}
			
			\begin{proof}
				Without loss of generality, we may assume that all parties $j$ with $L_i^j \notin C_i$ have the same local chain. Further, we can assume that this chain is just one block shorter than the currently longest chain. Thus, we search an upper bound for the probability that the parties $\{P_j\}_{L_i^j \notin C_i}$ are faster in solving two PoW's than the parties $\{P_j\}_{L_i^j \in C_i}$ solving one PoW.
				
				In order to prove that, we have to introduce a new random variable $\tilde{X_i}$, with  $\tilde{X_i} = 1$ if some honest party $j$ with $L_i^j \notin C_i$ solves a PoW. By the same argumentation as in Lemma \ref{lemma:successful-rounds-lemma}, we can argue that $ \frac{pq(1-s)E[n_{sleepy}]}{1 + pq(1-s)E[n_{sleepy}]} \leq E[\tilde{X_i}] \leq pq(1-s)E[n_{sleepy}] $. Therefore, the upper bound on the required probability is:
				
				\begin{align*}
			 &\sum_{k=2}^{\infty} (k-1) E[\tilde{X_i}]^2 (1-E[\tilde{X_i}])^{k-2} (1-E[X^*_i])^{k} \\
			= &\frac{E[\tilde{X_i}]^2}{(1-E[\tilde{X_i}])^2} \cdot \sum_{k=2}^{\infty} (k-1)\big((1-E[\tilde{X_i}])(1-E[X^*_i])\big)^k \\
			= & \frac{E[\tilde{X_i}]^2(1-E[X^*_i])^2}{(E[\tilde{X_i}] + E[X^*_i] - E[\tilde{X_i}]E[X^*_i])^2}
				\end{align*}
                
				Now, let $a:= pq(1-s)^2(n-t) = pq E[n^*_{alert}]$ and $b:= pqs(1-s)(n-t) = pq(1-s)E[n_{sleepy}] $. Then by the Definition of $E[\tilde{X_i}]$ and $E[X^*_i]$ holds:
				
				\begin{align*}
					\frac{E[\tilde{X_i}]^2(1-E[X^*_i])^2}{(E[\tilde{X_i}] + E[X^*_i] - E[\tilde{X_i}]E[X^*_i])^2} 
%					= &\frac{b^2 \big(\frac{1}{1+a}\big)^2}{\big(\frac{a}{1+a} + \frac{b}{1+b} - ab\big)^2} \\
					=& \frac{b^2}{(1+a)^2(1-ab)^2(a + ab +b)^2} \\
				\end{align*}
				
				Thus, $\phi \leq \frac{s}{1-s}$ is equivalent to:
				
				\begin{align*}
				        \frac{b^2}{(1+a)^2(1-ab)^2(a + ab +b)^2} \leq \frac{s}{1-s} \\
%		\Leftrightarrow	&b^2(1-s) 	&&\leq  s(1+a)^2(1-ab)^2(a + ab +b)^2 \\
		\Leftrightarrow ab         \leq (1+a)^2(1-ab)^2(a + ab +b)^2
				\end{align*}
				
				The inequality holds, since $(1+a)^2(1-ab)^2  \geq 1$ and $ab \leq (a + ab +b)^2$.  
			\qed\end{proof}

			The lemma below replaces Lemma 15 (of \cite{avarikioti2019bitcoin}). The possibility to have chains of different length at the same round offers various ways to replace a block from round $i$, where $Y^*_i = 1$. Thus, we cannot use the same arguments as in Lemma 15 (of \cite{avarikioti2019bitcoin}).
			
%----------------------------------IMPORTANT LEMMA-----------------	
			
			\begin{restatable}{lemma}{blockreplacement}
			\label{lemma:block-replacement-lemma}
			    Suppose the k$^{th}$ block $B$ of a chain $C$ was computed at round $i$, where $Y^*_i =1$. Then with probability at least $1-\phi$, the $k^{th}$ block in a chain $C'$ will be B or requires at least one adversarial block to replace $B$.
			\end{restatable}
			
% 			\begin{proof}
% 				There are several ways to replace a block from a round $i$, where $Y^*_i = 1$:
% 				\begin{itemize}
% 					\item[1)] The adversary has a precomputed block $B'$, replacing $B$ directly. Thus, in the same round, $\mathcal{A}$ diffuses the chain $C'$, where the last added block is $B'$.
					
% 					\item[2)] The parties $j$, with $L_i^j \notin C_i$ and thus with $L_i^j \neq C$ solve a PoW at some round $r \geq i$, leading to a new block on some chain $C'$, which has the same length as $C$. Then, $\mathcal{A}$ extends and diffuses $C'$ before $C$ gets further extended.
					
% 					\item[3)] The adversary solves a PoW and creates some chain $C'$. As in the second case, we may assume that the length of $C$ and $C'$ is equal. Then at some round $r \geq i$, either $\mathcal{A}$ or at least one party $j$ with $L_r^j = C'$ solve a PoW, resulting in the creation of the set $C_r$.
% 					\item[4)] The parties $j$ with $L_i^j \notin C_i$ are faster in solving two PoW's, before $C$ gets extended.
% 				\end{itemize}
				
% 				The cases 1) - 3) involve at least one adversarial block, as required. And by Lemma \ref{lemma:small-chance-lemma}, we know that the case 4) only happens with probability $\phi$.
% 			\qed\end{proof}
			
			As in the previous sections, the properties of the typical execution hold and executions are typical with high probability, by Lemma \ref{theorem:typical-execution-theorem}.
			
% 			\begin{definition}
% 			\label{definition:typical-execution-message-losses-definition}
% 				An execution is $(\epsilon, \eta)$-typical if, for any set $S$ of consecutive rounds with $|S| \geq \eta\kappa$, the following hold.
				
% 				\begin{itemize}
% 					\item[a)] $(1-\epsilon)E[X^*(S)] < X^*(S)$ and $X(S) < (1+\epsilon)E[X(S)]$
% 					\item[b)] $(1-\epsilon)E[Y^*(S)] < Y^*(S)$
% 					\item[c)] $Z(S) < (1+\epsilon)E[Z(S)]$
% 					\item[d)] No insertions, no copies and no predictions occurred.
% 				\end{itemize}
% 			\end{definition}			
			
% 			\begin{theorem}[\cite{GKL15}, Theorem 28]
% 				An execution is typical with probability $1-e^{-\Omega(\kappa)}$.
% 			\end{theorem}
			
% 			\begin{proof}
% 				Equivalent to the proof of Theorem \ref{theorem:typical-execution-theorem}.
				
% 			\qed\end{proof}
			
			Since we allow message losses in this model, we require more unique successful rounds than in other models. This leads to a different bound in part e) of the following lemma.
			
%--------LESS IMPORTANT LEMMA, BUT PART E) IS COMPLETELY NEW-----------------
			
			\begin{restatable}{lemma}{typicalexecutionc}
				\label{lemma:typical-execution-message-losses-lemma}
				The following hold for any set $S$ of at least $\eta\kappa$ consecutive rounds in a typical execution.
				
				\begin{itemize}
					\item[a)] $(1-\epsilon)E[X^*_i]|S| < X^*(S)$
					\item[b)] $(1-\epsilon)E[X^*_i](1-E[X^*_i])|S| < Y^*(S)$
					\item[c)] $Z(S) < (1+\epsilon)\frac{t}{E[n^*_{alert}]}\frac{E[X^*_i]}{1-E[X^*_i]}|S| < (1+\epsilon)\big(c(1-\delta)-\frac{s}{1-s}\big)\frac{E[X^*_i]}{1-E[X^*_i]}|S|$
					\item[d)] For $\sigma^* = (1-\epsilon)(1-E[X^*_i])$:
					
					$$Z(S) < \Big(1 + \frac{\delta}{\sigma^*}\Big) \frac{t}{E[n^*_{alert}]} X^*(S) \leq c\Big(1-\frac{\delta^2}{2\sigma^*}\Big) X^*(S)$$
					\item[e)] $$Z(S) < Y^*(S)(1-\epsilon)(1-\phi)$$
				\end{itemize}
			\end{restatable}

Next, we prove Bitcoin is secure in the synchronous model with message loss. The proof can be found in the full version.

\section{Security Analysis Results}

    As a result of the temporary dishonest majority assumptions, we have derived upper bounds for the probability $s$ as shown in Figure \ref{figure:combined}. Therefore, we fixed $c=0.5$ to limit the advantage of an adversary, following a Selfish Mining strategy. Further, we have chosen for all three models $\epsilon = 0.005$. For the synchronous model without message losses, we set $E[X_i] = 0.03$, which results in $\delta = 0.07$\footnote{Note that $\delta$ is dependent on $E[X_i]$, which is again dependent on $s$. If we would remove this dependency, the results would be at most 2\% better than the actual results shown in Figure \ref{figure:combined}.}.
    For the Semi-Synchronous model, we set also $E[X_i] = 0.03$, resulting in $E[X'_i] = 0.022$. For $\Delta = 10$, we then get $\delta = 0.46$. And for the synchronous model with message losses, we have chosen $E[X^*_i] = 0.03$, which results in $\delta = 0.075$.

    \begin{figure}
    \centering
    \label{figure:combined}
        \includegraphics[scale=0.6]{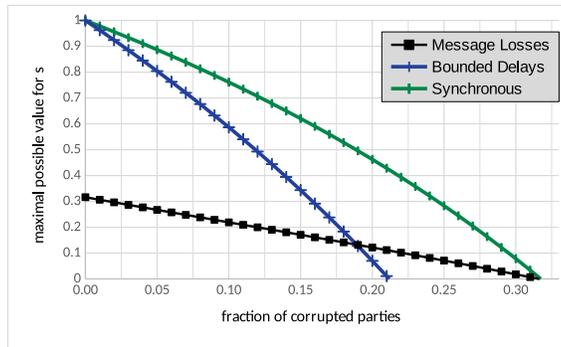}
        \caption{This figure shows the upper bound on the fraction of sleepy parties, depending on the fraction of corrupted parties.} 
    \end{figure}

    One might be wondering how we could allow such high values for $s$. We have fixed $E[X_i]$, respectively $E[X^*_i]$, for our calculations. We can do this without loss of generality, since these expected values are dependent on $p$, which depends on the difficulty parameter $T$. The adjustment of $T$, used to regulate the block generation rate, depends on the fraction of sleepy parties, because they do not provide computational power (e.g. new blocks) to the blockchain.
    
    These results are also consistent with the results from \cite{DW13}, where the upper bound on the adversarial fraction is stated at $49.1\%$. If we set $c=1$ and $s=0$, due the value of $\delta$, we get an maximal possible adversarial fraction of $48.5\%$.
	
\section{Related Work}
    %In a seminal work, Garay et al.~\cite{GKL15} proposed a formal cryptographic protocol, called the backbone protocol, to describe the Bitcoin system. Furthermore, in the same work, the security properties for the backbone protocol were defined and proven both in the synchronous and bounded-delay model. However, in contrast to our model, they assumed honest majority of the participants. In this work, we relaxed this assumption allowing the adversary to temporary control the majority of the mining power of the system, as observed in practice. 
    
    To model temporary dishonest majority in Bitcoin we used an idea, originally introduced by  Pass and  Shi \cite{PS17}. In this work, they introduced the notion of sleepy nodes, i.e. nodes that go offline during the execution of the protocol, and presented a provably secure consensus protocol. In this paper, we model the dynamic nature of the system by additionally allowing the adversary to set parties to sleep, thus enabling temporary dishonest majority.
    
    Bitcoin has been studied from various aspects and multiple attacks have been proposed, concerning the network layer \cite{HKZG15eclipse,Nayak2015StubbornMG,AZV17} as well as the consensus algorithm (mining attacks) \cite{ES13,SSZ16,Eyal15dillema,KKSV17FAW}. The most famous mining attack is selfish mining \cite{ES13}, where a selfish miner can withhold blocks and gain disproportionate revenue compared to his mining power. The chain quality property, originally introduced in \cite{GKL15}, encapsulates this ratio between the mining power and the final percentage of blocks, and thus rewards, the adversary owns. 
    On the other hand, Heilman et al.~\cite{HKZG15eclipse} examined eclipse attacks on the Bitcoin's peer-to-peer network. In turn, Nayak et al.~\cite{Nayak2015StubbornMG} presented a novel attack combining selfish mining and eclipse attacks. They showed that in some adversarial strategies the victims of an eclipse attack can actually benefit from being eclipsed. Our last model, where offline parties do not get the update messages, captures this attack.
    % This work is built upon the ideas and results of two papers. The Backbone protocol and the framework for the security proofs can be found in \cite{GKL15}. Further, a large part of the model is adapted from \cite{GKL15}, which is founded on the work of Ran Canetti in \cite{Can00a,Can00b,Can01}. The sleepy model was introduced in \cite{PS17}, in which Rafael Pass and Elaine Shi present a provably secure protocol when parties can go offline for a certain amount of time. But the proofs are based on the assumption that the fraction of corrupted over the alert parties is always less than some constant. As we showed, we can even relax this assumption for Bitcoin by just assuming honest majority on expectation.
%
    % Ouroboros, described in \cite{KRDO17}, is also proven in a sleepy model. The parties, called stakeholders, can be offline for some time, but are required to be online at important events such as during the slot, where they are elected slot leaders. Since in Ouroboros, the parties receive their block reward when they are elected slot leader, they have a much higher interest in being online at this certain time slot, where in Bitcoin the overall online time (combined with the hashing power) determines the rewards. This is then captured in our model, where each player is round independent set to sleep with some probability.
%
\section{Conclusion \& Future Work}
    In this paper, we prove Bitcoin is secure under temporary dishonest majority.
    Specifically, we extended the framework of Garay et al.~\cite{GKL15} to incorporate offline nodes and allow the adversary to introduce crush failures. This way we can relax the honest majority assumption and allow temporary dishonest majority. We prove Bitcoin' s security by showing that under an expected honest majority assumption the following security properties hold: chain growth, common prefix and chain quality. 
    
    We examine three models: the synchronous model, the bounded delay model and the synchronous model with message loss. The first two models result in similar bounds regarding the fractions of corrupted and sleepy parties. In contrast, the last model that allows message losses when a party goes offline is less resilient to sleepy behavior. This is expected since this model captures the nature of eclipse attacks where the adversary can hide part of the network form an honest party and either waste or use to his advantage the honest party's mining power. We illustrate in Figure \ref{figure:combined} the upper bounds on the fraction of sleepy parties depending on the fraction of corrupted parties for all three models.
    
    For future work, we did not consider the bounded delay with message loss model. We expect the difference on the results from synchronous to bounded delay model to be similar to the model without message loss. Another interesting future direction is to consider a more powerful adversary, who knows whether his attempt to set a party to sleep is successful or not.  
    
    % In this paper we combined and extended the work of \cite{GKL15} and \cite{PS17}. We were able to prove the security of Bitcoin, respectively the security of the Backbone protocol in the sleepy model, by just assuming honest majority on expectation. This shows the resistance of the Bitcoin protocol, which is able to recover from peaks of adversarial power above $50\%$.
    % The q-bounded synchronous model, where we allow message losses provide one of the most realistic abstractions of the reality and shows nevertheless the security of Bitcoin. The results imply that Bitcoin can withstand very powerful adversaries who can perform eclipse attacks on a number of honest parties.
    
    % However, we did not consider a more powerful adversary, who is informed about the success of introduced crash-failures. This would be very interesting to study, since this involves a non-constant probability $s$ and new strategy for the adversary. 

\section{Acknowledgments}
We thank Dionysis Zindros for the helpful and productive discussions. Y.\ W.\ is partially supported by X-Order Lab.

% \newpage
\bibliographystyle{splncs04}
\bibliography{references}

\end{document}